\documentclass[11pt]{article}

\usepackage[a4paper,margin=25mm]{geometry}
\usepackage{graphicx,xcolor}
\usepackage{float}
\usepackage{amsmath,amssymb,bm,amsthm}
\usepackage{array}
\usepackage{hyperref}
\hypersetup{colorlinks=true,linkcolor=red,citecolor=blue,urlcolor=orange}
\usepackage[hang,small,bf]{caption}
\usepackage[subrefformat=parens]{subcaption}
\usepackage{natbib}
\usepackage{enumitem}
\usepackage[framemethod=tikz]{mdframed}
\usepackage{placeins}

\newmdenv[linewidth=0.4pt,everyline=true,skipabove=6pt,skipbelow=6pt,
  innerleftmargin=5pt,innerrightmargin=5pt,
  innertopmargin=3pt,innerbottommargin=3pt,
  splittopskip=14pt,splitbottomskip=4pt]{paperbox}
\title{A Second-Order Asymptotic Expansion for Nonlinear Filtering in the Lorenz--63 System}
\author{Masahiro Kurisaki\thanks{AIP Center, RIKEN (\href{mailto:masahiro.kurisaki@riken.jp}{\texttt{masahiro.kurisaki@riken.jp}}).}
\and Kota Takeda\thanks{Department of Applied Physics, Graduate School of Engineering,
Nagoya University, Nagoya, Japan (\href{mailto:takeda@na.nuap.nagoya-u.ac.jp}{\texttt{takeda@na.nuap.nagoya-u.ac.jp}}).}}
\date{October 7, 2026}
\numberwithin{equation}{section}
\allowdisplaybreaks[3]
\newtheorem{theorem}{Theorem}[section]
\newtheorem{proposition}[theorem]{Proposition}

\theoremstyle{definition}

\theoremstyle{remark}
\newtheorem{remark}{Remark}

\begin{document}
\maketitle

\begin{abstract}
  A challenge in nonlinear filtering is to incorporate nonlinear effects into approximations of the conditional mean while retaining a finite-dimensional recursive algorithm. We derive a second-order asymptotic expansion of the filtering mean with respect to the system-noise intensity. We apply this expansion to the Lorenz--63 model, a prototypical three-dimensional chaotic dynamical system. The first-order approximation is computed using a Kalman--Bucy filter for the system linearized around the deterministic reference trajectory. The second-order correction incorporates the Hessian of the Lorenz--63 drift and is computed through a closed recursive system with $3+9+18=30$ independent scalar components. The full system, including the reference trajectory and the first-order mean and covariance, has 42 independent scalar components. To support stable tracking of chaotic trajectories over long time intervals, we introduce recentering of the expansion and redistribution of Gaussian components between expansion orders. In the numerical experiments, we compare the proposed algorithm with an ensemble transform Kalman filter (ETKF) as a practical benchmark and a particle filter (PF) as a numerical reference. Under the full-observation settings, the proposed algorithm achieves a lower root-mean-square error relative to the truth and a smaller distance to the 10,000-particle PF reference than the 20-member ETKF. These results suggest that the proposed second-order extension of Kalman filtering has the potential to balance filtering accuracy and computational cost.
\end{abstract}

\section{Introduction}

Sequential state estimation for nonlinear dynamical systems arises in many geophysical applications, including numerical weather prediction \citep{evensenDataAssimilationFundamentals2022}. Bayesian nonlinear filtering describes this problem through the conditional distribution of the state given the observation history. Its conditional mean minimizes the conditional mean-square estimation error. In continuous time, the conditional distribution is measure-valued and is characterized, under suitable assumptions, by the Kushner--Stratonovich equation \citet{kushner1964,stratonovich1960}; an unnormalized version satisfies the Zakai equation \citep{Mortensen1966,Zakai1969,Duncan1970a,Duncan1970b}. These equations are generally infinite-dimensional \citep{HAZEWINKEL1983}. For linear systems with Gaussian initial conditions and Gaussian noise, however, the conditional distribution remains Gaussian and the Kalman--Bucy filter gives closed equations for its mean and covariance \citep{kalman1960,kalman1961}. Such a finite-dimensional closure is available only for special classes of nonlinear models. This motivates finite-dimensional approximations that retain relevant nonlinear effects on the conditional distribution, and in particular on its mean.

Particle filters approximate the conditional distribution by weighted empirical measures and converge as the number of particles tends to infinity under standard regularity assumptions \citep{gordon1993,Crisan2006}. Their direct importance-sampling implementations can nevertheless suffer from weight degeneracy. In representative high-dimensional settings, preventing this degeneracy requires a number of particles that grows exponentially with an effective dimension \citep{snyderObstaclesHighDimensionalParticle2008,bickelSharpFailureRates2008,suraceHowAvoidCurse2019}. Continuous-time interacting-particle formulations provide another class of nonlinear filtering methods \citep{taghvaeiKalmanFilterIts2018}. Ensemble Kalman methods instead update an ensemble using its sample mean and covariance, and are widely used in large-scale geophysical applications because they can be effective with much smaller ensembles \citep{evensenDataAssimilationEnsemble2009}. Stability and accuracy results are available under specific assumptions, for example for the ensemble Kalman--Bucy filter with full observations and small observation noise and for the ensemble transform Kalman filter with covariance inflation \citep{dewiljesLongTimeStability2018,takedaUniformErrorBounds2024}. In the linear-Gaussian setting, ensemble Kalman methods converge to the Kalman filter as the ensemble size increases \citep{mandelConvergenceEnsembleKalman2011,kwiatkowskiConvergenceSquareRoot2015}. For a general nonlinear filtering problem, however, this large-ensemble limit is not the exact nonlinear filter.

An earlier asymptotic approach was developed by \citet{picardNonlinearFilteringOneDimensional1986} for one-dimensional diffusion processes in a small-observation-noise limit. This work constructs finite-dimensional approximate filters from asymptotic estimates of conditional moments and establishes approximation error bounds. The bound requires an injective observation function, which restricts its direct applications to partially observed systems.
Another asymptotic approach expands the Zakai equation for the characteristic function (the Fourier transformation) of the unnormalized filtering distribution \citep{fujiiMomentumspaceApproachAsymptotic2014}. The computational cost can exponentially increase depending on the state dimension if we assign an equal number of the Fourier modes for each component. Moreover, the expansion has limitations; it is inconsistent because the observation function goes to $0$ as the asymptotic parameter $\epsilon$ goes to $0$; no rigorous bound is obtained for the truncation remainder.

\citet{kurisaki2026ks} rigorously derives a small-system-noise expansion of conditional expectations and provides a probabilistic bound on the truncation remainder.
\citet{kurisaki2026recursions} shows that the expansion coefficients can be computed through a finite-dimensional system of stochastic differential equations driven by the observation process.
For fixed expansion order and representation structure, the dimension of this system grows at most polynomially with the state dimension.
Explicit recursive equations are worked out through third order in a one-dimensional setting, while deriving tractable equations for specific multidimensional models requires further calculation.

In this paper, we derive an explicit form of the second-order small-system-noise expansion of \citet{kurisaki2026ks,kurisaki2026recursions} for a stochastically forced Lorenz--63 model with continuous-time linear observations. The Lorenz--63 model is widely used as a three-dimensional chaotic model related to geophysical applications \citep{lorenzDeterministicNonperiodicFlow1963}. Our aim is to compute the conditional mean beyond the linearized approximation while retaining a finite-dimensional recursive algorithm.
The first-order coefficient is obtained from the Kalman--Bucy filter associated with the first-variation system. At second order, the recursive state contains coefficient arrays of sizes $d$, $d^2$, and $d^3$ for a $d$-dimensional signal; thus the number of stored coefficient components is $O(d^3)$. This estimate concerns the size of the coefficient system, not a complete runtime comparison with particle filters or asymptotic filters. For long-time filtering, we also introduce reparameterizations of the expansion through recentering and redistribution of Gaussian components between expansion orders. The numerical experiments reported below indicate that these reparameterizations substantially improve long-time stability, although their theoretical justification is outside the scope of this work. In the tested small-to-moderate system-noise regimes, the resulting second-order estimate attains accuracy close to that of the large-ensemble bootstrap particle-filter reference; we do not claim this behavior uniformly over noise levels or model classes.
Our expression of the expansion still has limitations when applied to partially observed high-dimensional systems.

The expansion is numerically unstable under one-dimensional observations of the Lorenz--63 model. Moreover, the coefficient count of $\mathcal{O}(d^3)$ is still large for high-dimensional systems. Accordingly, deriving a more stable representation and a lower-dimensional approximation of our expansion is left for future work.

The contributions of this work are as follows:
\begin{enumerate}
\item We specialize the small-system-noise expansion of \citet{kurisaki2026ks,kurisaki2026recursions} to the continuously observed Lorenz--63 model and derive a finite-dimensional recursive system for the first- and second-order coefficients of the conditional mean.
\item We introduce recentering and Gaussian-component redistribution, and show numerically that these transformations yield more stable long-time computations than the uncentered representation.
\item We compare the second-order approximation with a particle filter and an ensemble transform Kalman filter over several system-noise levels. In the tested small-to-moderate noise regimes, the second-order method gives estimation accuracy close to that of the particle-filter reference while evolving a coefficient system with $O(d^3)$ components.
\end{enumerate}

The remainder of this paper is organized as follows. Section 2 presents the second-order expansion and derives recursive equations for its coefficients. Section 3 specializes these equations to the Lorenz--63 model. Section 4 presents numerical experiments and introduces reset operations for long-time computations. Section 5 summarizes the findings and discusses future directions.

\section{General theory for the second-order expansion of the nonlinear filter}
\subsection{Expansion of the system}
Let $(\Omega,\mathcal{F},\{\mathcal{F}_t\}_{t\geq 0},P)$ be a filtered probability space, and assume $d_1$ and $d_2$-dimensional stochastic processes $\{X_{t}^\epsilon\}_{t\geq 0}$ and $\{Y_t^\epsilon\}$ satisfy the stochastic differential equation
\begin{align}
  \label{eq-X}&dX_t^\epsilon=a(X_t^\epsilon)dt+\epsilon dV_t,~~X_0=x_0,\\
  \label{eq-Y}&dY_t^\epsilon=HX_t^\epsilon dt+\sigma dW_t,~~Y_0=0.
\end{align}
Here, $a:\mathbb{R}^{d_1}\to \mathbb{R}^{d_1}$ is a function of class $C^\infty$, $H \in \mathbb{R}^{d_2\times d_1}$, $\epsilon>0$, $\sigma>0$ and $x_0 \in \mathbb{R}^{d_1}$ are constants, and $\{V\}_{t\geq 0}$ and $\{W\}_{t \geq 0}$ are independent $d_1$ and $d_2$-dimensional $\{\mathcal{F}_t\}_{t \geq 0}$-Wiener processes.
The deterministic initial state is assumed to simplify the derivation;
random initial coefficients relevant to restarting are discussed in
Remark~\ref{rem:random-initial-coefficients}.

Furthermore, we assume that for any $k \in \mathbb{Z}_+$
\begin{align}
  \label{eq:assumption-a}\sup_{x \in \mathbb{R}^{d_1}}|\partial^ka(x)|<\infty
\end{align}
for some constant $C>0$. Here, for $\phi \in C^k(\mathbb{R}^m; \mathbb{R}^n)$, the $k$-th derivative $\partial^k \phi(x)$ is regarded as an $\mathbb{R}^n$-valued $k$-tensor whose $(i_1,\ldots,i_k)$-th component is given by
  \begin{align*}
      \frac{\partial}{\partial x_{i_1}} \cdots \frac{\partial}{\partial x_{i_k}} \phi(x_1, \cdots, x_m).
  \end{align*}

Under this condition, the system $\{X_t^\epsilon\}_{t\geq 0}$ has the following Taylor expansion.
\begin{proposition}[Propositon 3.1 in \citet{kurisaki2026ks}]
  For any $t\geq 0$, we have the expansion
  \begin{align*}
    X_t^\epsilon=\sum_{k=0}^{n} \frac{\epsilon^k}{k!}X_t^{[k]}+r_t^{n+1,\epsilon},
  \end{align*}
  where $X_t^{[k]}$ is the solution of the $k$-th formal derivative of the system equation (\ref{eq-X}) at $\epsilon=0$, and the remainder term $r_t^{\epsilon,n+1}$ satisfies
  \begin{align*}
    \sup_{0<\epsilon<1}\frac{1}{\epsilon^{n+1}}E\left[\sup_{0\leq s\leq t} \left|r_s^{\epsilon,n+1}\right|^p \right]^\frac{1}{p}<\infty.
  \end{align*}
  for any $t\geq 0$ and $p\geq 1$.

  More explicitly, $X_t^{[k]}$ is the solution of the equation
  \begin{align*}
    &dX_t^{[0]}=a(X_t^{[0]})dt\\
    &dX_t^{[1]}=\partial a(X_t^{[0]})X_t^{[1]}dt+dV_t,\\
  &dX_t^{[k]}=\sum_{j=1}^k \sum_{(i_1,\cdots,i_j)\in \Lambda(k,j)}\nu(k,(i_1,\cdots,i_j))\partial^{j}a(X_t^{[0]})[X_t^{[i_1]}\otimes \cdots \otimes X_t^{[i_j]}]dt,\qquad k\geq 2,
\end{align*}
with $X_0^{[0]}=x_0$ and $X_0^{[k]}=0$ for $k\geq 1$. Here,
\begin{align*}
  \Lambda(k,j)=\left\{ (i_1,\cdots,i_j)\in \mathbb{N}^j;i_1\leq \cdots \leq i_j,~i_1+\cdots+i_j=k  \right\},
\end{align*}
and $\nu(k,(i_1,\cdots,i_j))$ represents the number of partitions of the set $\{1,\cdots,k\}$ into $j$ subsets with $(i_1,\cdots,i_j)$ elements, and
\begin{align*}
  &\partial^{j}a(x_1,\cdots,x_{d_1})[\xi^1\otimes \cdots \otimes \xi^j]\\
  =&\sum_{i_1,\cdots,i_j=1}^{d_1}\frac{\partial}{\partial x_{i_1}} \cdots \frac{\partial}{\partial x_{i_j}} a(x_1, \cdots, x_{d_1})\xi_{i_1}^1\cdots \xi_{i_j}^j
\end{align*}
for $\xi^{i}=(\xi_1^i,\cdots,\xi_{d_1}^i)~(i=1,\cdots,j)$.
\end{proposition}

In this paper, we focus on the second-order expansion and only use the case of $k=2$. In this case, the equation for $X_t^{[2]}$ is written as
\begin{align}
  dX_t^{[2]}=\partial a(X_t^{[0]})X_t^{[2]}dt+\partial^2 a(X_t^{[0]})[(X_t^{[1]})^{\otimes 2}]dt.\label{eq:X2}
\end{align}

\begin{remark}
  The Lorenz--63 drift does not satisfy assumption (\ref{eq:assumption-a}).
  For the theoretical argument above, one may introduce a smooth cutoff outside
  a sufficiently large region of the state space. The resulting drift satisfies
  the bounded-derivative assumption, while agreeing with the Lorenz--63 drift on
  the region relevant for the finite-time dynamics under consideration.
\end{remark}

\subsection{Second-order expansion of the Kallianpur--Striebel formula}
In this section, we consider the expansion of the filtering mean $E[X_t^\epsilon|\mathcal{Y}_t^\epsilon]$, where $\{\mathcal{Y}_t^\epsilon\}_{t\geq 0}$ is the filtration generated by $\{Y_t^\epsilon\}_{t\geq 0}$ augmented by null sets. To this end, we first recall the explicit expression of the conditional expectation given by the Kallianpur--Striebel formula.

Fix $T>0$ and define a probability measure $Q^\epsilon$ on $\mathcal{F}$ by
\begin{align*}
  Q^\epsilon(A) = E\left[ 1_A (Z_T^\epsilon)^{-1} \right]
\end{align*}
for \( A \in \mathcal{F} \), where
\begin{align}
  \begin{split}
    Z_t^\epsilon &= \exp\left( \frac{1}{\sigma^2}\int_0^t (HX_s^\epsilon)^\top  dY_s^\epsilon- \frac{1}{2\sigma^2} \int_0^t |HX_s^\epsilon|^2 ds \right).
  \end{split}\label{eq-exp-Z}
\end{align}

Since
\begin{align*}
  \int_0^t E[|X_s|^2]ds<\infty,
\end{align*}
$Z_t^\epsilon$ is a martingale and $Q^\epsilon$ is well-defined as a probability measure.

\begin{proposition}[Kallianpur--Striebel formula \citep{bain_crisan2009filtering}]
  If a random variable $U$ is independent of $W$ under $P$ and \( E[|U|] < \infty \), then it holds that
\begin{align}
  \label{eq-Kallianpur-Striebel-epsilon}
  E[U|\mathcal{Y}_t^\epsilon] =
  \frac{\displaystyle E_{Q^\epsilon}\left[ U Z_t^\epsilon \middle| \mathcal{Y}_t^\epsilon \right]}
       {\displaystyle E_{Q^\epsilon}\left[ Z_t^\epsilon \middle| \mathcal{Y}_t^\epsilon \right]}
  \quad (0 \leq t \leq T).
\end{align}
\end{proposition}

In order to consider the expansion of the formula, write
\begin{align*}
  \tilde{Z}_t^\epsilon=\exp\left( \frac{1}{\sigma^2}\int_0^t (H(X_s^{[0]}+\epsilon X_s^{[1]}))^\top  dY_s^\epsilon- \frac{1}{2\sigma^2} \int_0^t |H(X_s^{[0]}+\epsilon X_s^{[1]})|^2 ds \right).
\end{align*}
Then, since
\begin{align*}
  X_t^\epsilon=X_t^{[0]}+\epsilon X_t^{[1]}+\frac{\epsilon^2}{2}X_t^{[2]}+O(\epsilon^3),\quad Z_t^\epsilon=\tilde{Z}_t^\epsilon+O(\epsilon^2),
\end{align*}
we have
\begin{align*}
  E_{Q^\epsilon}[ (X_t^\epsilon-X_t^{[0]}) Z_t^\epsilon | \mathcal{Y}_t^\epsilon ]
  =\epsilon E_{Q^\epsilon}[ X_t^{[1]}\tilde{Z}_t^\epsilon | \mathcal{Y}_t^\epsilon ]
  +\frac{\epsilon^2}{2} E_{Q^\epsilon}[ X_t^{[2]}\tilde{Z}_t^\epsilon | \mathcal{Y}_t^\epsilon ] + O(\epsilon^3)
\end{align*}
and
\begin{align*}
  E_{Q^\epsilon}[ Z_t^\epsilon | \mathcal{Y}_t^\epsilon ]
  =E_{Q^\epsilon}[ \tilde{Z}_t^\epsilon | \mathcal{Y}_t^\epsilon ]+O(\epsilon^2).
\end{align*}
Therefore, we have
\begin{align*}
  E[X_t^\epsilon|\mathcal{Y}_t^\epsilon]&=X_t^{[0]}+E[X_t^\epsilon-X_t^{[0]}|\mathcal{Y}_t^\epsilon]
  =X_t^{[0]}+\frac{E_{Q^\epsilon}[ (X_t^\epsilon-X_t^{[0]}) Z_t^\epsilon | \mathcal{Y}_t^\epsilon ]}{E_{Q^\epsilon}[ Z_t^\epsilon | \mathcal{Y}_t^\epsilon ]}\\
  &=X_t^{[0]}+\epsilon \frac{E_{Q^\epsilon}[X_t^{[1]} \tilde{Z}_t^\epsilon | \mathcal{Y}_t^\epsilon ]}{E_{Q^\epsilon}[ \tilde{Z}_t^\epsilon | \mathcal{Y}_t^\epsilon ]}+\frac{\epsilon^2}{2} \frac{E_{Q^\epsilon}[X_t^{[2]} \tilde{Z}_t^\epsilon | \mathcal{Y}_t^\epsilon ]}{E_{Q^\epsilon}[ \tilde{Z}_t^\epsilon | \mathcal{Y}_t^\epsilon ]}+O(\epsilon^3).
\end{align*}
The rigorous justification of this argument is given in Theorem 4.10 of \citet{kurisaki2026ks}.

\begin{theorem}[\citep{kurisaki2026ks}]\label{theorem:expansion-KS-formula}
  For any $t\geq 0$, we can write
  \begin{align*}
    E[X_t^\epsilon|\mathcal{Y}_t^\epsilon]=X_t^{[0]}+\epsilon \tilde{E}_t[X_t^{[1]}]+\frac{\epsilon^2}{2}\tilde{E}_t[X_t^{[2]}]+O_P^T(\epsilon^3),
  \end{align*}
  where
  \begin{align*}
    \tilde{E}_t[U]=\frac{E_{Q^\epsilon}[U \tilde{Z}_t^\epsilon | \mathcal{Y}_t^\epsilon ]}{E_{Q^\epsilon}[ \tilde{Z}_t^\epsilon | \mathcal{Y}_t^\epsilon ]}
  \end{align*}
  for $U \in L^1(P)$, and if a family of random variables $\{\eta_t^\epsilon\}_{0\leq \epsilon<1,t\geq 0}$ is $O_P^T(\epsilon^k)$, then
  \begin{align*}
    \sup_{0\leq t\leq T}\sup_{0<\epsilon<1}P(|\eta_t^\epsilon|> K\epsilon^k)\to 0~~(K\to \infty).
  \end{align*}
\end{theorem}
\begin{proof}
See the construction of Section 3 in \citet{kurisaki2026ks}, and apply Theorem 4.10
with $n=3$, setting
\begin{align*}
  &a(t)=\partial a(X_t^{[0]}),\qquad b(t)=I_{d_1},\qquad
 c_0(t)=HX_t^{[0]},\qquad c(t)=\epsilon H,\qquad
 \sigma(t)=\sigma I_{d_2},\\
 &H_t^{[1]}=0,\quad H_t^{[2]}=\tfrac12 HX_t^{[2]},\\
 &U_t^{[0]}=X_t^{[0]},\quad U_t^{[1]}=X_t^{[1]},\quad U_t^{[2]}=\tfrac12 X_t^{[2]}.
\end{align*}
\end{proof}

\subsection{First order equation}
Define the conditional mean and covariance of the first variation by
\begin{align*}
  \mu_t&=\tilde E_t[X_t^{[1]}],\\
  \gamma(t)&=\tilde E_t[
    (X_t^{[1]}-\mu_t)(X_t^{[1]}-\mu_t)^\top].
\end{align*}
Their dependence on $\epsilon$ through the linearized likelihood is
suppressed in the notation.

\begin{proposition}[\citep{KURISAKI2026104997}]
Versions of $\mu_t$ and $\gamma(t)$ satisfy the Kalman--Bucy equations
\begin{align*}
  d\mu_t
  &=
  \partial a(X_t^{[0]})\mu_t\,dt
  +\frac{\epsilon}{\sigma^2}\gamma(t) H^\top
  \left\{
    dY_t^\epsilon-H(X_t^{[0]}+\epsilon\mu_t)\,dt
  \right\},\\
  \frac{d}{dt}\gamma(t)
  &=
  \partial a(X_t^{[0]})\gamma(t)
  +\gamma(t)\partial a(X_t^{[0]})^\top
  +I_{d_1}
  -\frac{\epsilon^2}{\sigma^2}\gamma(t) H^\top H\gamma(t),
\end{align*}
with initial conditions $\mu_0=0$ and $\gamma(0)=0$.
\end{proposition}
The rigorous proof of this statement is exactly the same as the derivation of the Kalman--Bucy filter from the Kallianpur--Striebel formula given in \citet{KURISAKI2026104997}.

\subsection{Second-order coefficient}
\label{sec:second-order-coefficient}
We derive a recursion for $m_t^{(2)}=\tilde E_t[X_t^{[2]}]$ by
differentiating its integral representation with respect to the terminal
observation time. A related derivation of the second-order coefficients in
the one-dimensional setting is given in
\citet[Section 3]{kurisaki2026recursions}.

We identify $u\otimes v$ with $uv^\top$ and write $u^{\otimes2}=u\otimes u$.
For a matrix $M$, Hessian contraction means
\[
  \bigl(\partial^2a(x)[M]\bigr)_\ell
  =\sum_{p,q=1}^{d_1}\partial_p\partial_q a_\ell(x)M_{pq}.
\]
In particular, $\partial^2a(x)[u\otimes v]$ is the usual bilinear Hessian action.
The fundamental matrix is defined by
\begin{equation}
  \frac{\partial}{\partial t}\Phi(s,t)=\partial a(X_t^{[0]})\Phi(s,t),\qquad \Phi(s,s)=I_{d_1}.
  \label{eq:second-fundamental}
\end{equation}
By \eqref{eq:X2},
\begin{equation}
  X_t^{[2]}=\int_0^t\Phi(s,t)\partial^2a(X_s^{[0]})
                   [(X_s^{[1]})^{\otimes2}]\,ds.
\end{equation}
Let
\begin{equation}
  \mu_{s;t}=\tilde E_t[X_s^{[1]}],\qquad
  \gamma(s,u;t)=\operatorname{Cov}_{\tilde E_t}(X_s^{[1]},X_u^{[1]}),
\end{equation}
so that $\mu_{t;t}=\mu_t$ and $\gamma(t,t;t)=\gamma(t)$.
Conditional Fubini's theorem and the Gaussian second-moment formula give
\begin{equation}
  m_t^{(2)}=\int_0^t\Phi(s,t)\partial^2a(X_s^{[0]})
                    [\mu_{s;t}^{\otimes2}+\gamma(s,s;t)]\,ds.
  \label{eq:second-mean-integral}
\end{equation}
In order to differentiate this integral with respect to $t$, we use the following formulas for fixed $s,u<t$.
Here $d_t$ denotes the It\^o differential with respect to $t$.
\begin{align}
  d_t\mu_{s;t}&=\frac{\epsilon}{\sigma^2}\gamma(s,t;t)H^\top \left\{ dY_t^\epsilon-H(X_t^{[0]}+\epsilon\mu_t)\,dt \right\},\label{eq:smoothing1}\\
  \frac{\partial}{\partial t}\gamma(s,u;t)&=-\frac{\epsilon^2}{\sigma^2}\gamma(s,t;t)H^\top H\gamma(t,u;t),\label{eq:smoothing2}\\
  \frac{\partial}{\partial t}\gamma(s,t;t)&=\gamma(s,t;t)\left\{\partial a(X_t^{[0]})^\top-\frac{\epsilon^2}{\sigma^2}H^\top H\gamma(t)\right\},
  \label{eq:smoothing3}
\end{align}
where \eqref{eq:smoothing1} means 
\begin{align*}
  \mu_{s;t}=\mu_{s;s}+\int_{s}^t\frac{\epsilon}{\sigma^2}\gamma(s,r;r)H^\top \left\{ dY_r^\epsilon-H(X_r^{[0]}+\epsilon\mu_r)\,dr \right\}.
\end{align*}
For the proof of these equations, see \citet{KURISAKI2026104997}.

In what follows, we write
\begin{equation}
  d\nu_t=dY_t^\epsilon-H(X_t^{[0]}+\epsilon\mu_t)\,dt.
  \label{eq:second-innovation}
\end{equation}

\paragraph{The equation for $m_t^{(2)}$}
Using \eqref{eq:smoothing1}, \eqref{eq:smoothing2} and It\^o's formula, we can prove that
\begin{equation}
\begin{aligned}
 d_t\{\mu_{s;t}^{\otimes2}+\gamma(s,s;t)\}
 &=\frac{\epsilon}{\sigma^2}
   \{\mu_{s;t}\otimes(\gamma(s,t;t)H^\top d\nu_t)\\
 &\qquad+(\gamma(s,t;t)H^\top d\nu_t)\otimes\mu_{s;t}\}.
\end{aligned}
\label{eq:second-N-differential}
\end{equation}
Therefore, differentiating \eqref{eq:second-mean-integral} gives
\begin{equation}
 dm_t^{(2)}
 =\{\partial a(X_t^{[0]})m_t^{(2)}+\partial^2a(X_t^{[0]})
       [\mu_t^{\otimes2}+\gamma(t)]\}\,dt
   +\frac{\epsilon}{\sigma^2}Q_t^\top H^\top d\nu_t,
 \label{eq:second-m-equation}
\end{equation}
where
\begin{equation}
 Q_t^\top
 =2\int_0^t\Phi(s,t)\partial^2a(X_s^{[0]})[\mu_{s;t}]
                 \gamma(s,t;t)\,ds.
 \label{eq:second-Q-integral}
\end{equation}
Here $\partial^2a(x)[u]$ denotes the matrix satisfying
$\partial^2a(x)[u]v=\partial^2a(x)[u\otimes v]$.
All integrals being differentiated are with respect to $ds$;
the interchange with stochastic integration is justified by stochastic
Fubini under the finite-horizon moment bounds.

\paragraph{The equation for $Q_t$}
For fixed $s$, \eqref{eq:smoothing1} and \eqref{eq:smoothing2} yields
\begin{equation}
\begin{aligned}
 d_t\{\partial^2a(X_s^{[0]})[\mu_{s;t}]\gamma(s,t;t)\}
 &=\frac{\epsilon}{\sigma^2}\partial^2a(X_s^{[0]})
       [\gamma(s,t;t)H^\top d\nu_t]\gamma(s,t;t)\\
 &\quad+\partial^2a(X_s^{[0]})[\mu_{s;t}]\gamma(s,t;t)
       \left(\partial a(X_t^{[0]})^\top-\frac{\epsilon^2}{\sigma^2}H^\top H\gamma(t)\right)dt.
\end{aligned}
\end{equation}
Define the symmetric bilinear tensor $K(t)$ by
\begin{equation}
 K(t)[u\otimes v]
 =2\int_0^t\Phi(s,t)\partial^2a(X_s^{[0]})
       [(\gamma(s,t;t)u)\otimes(\gamma(s,t;t)v)]\,ds,
 \label{eq:second-K-integral}
\end{equation}
and write $K(t)[u]v=K(t)[u\otimes v]$.
Its components are $K_{\ell ij}(t)=(K(t)[e_i\otimes e_j])_\ell$,
where $e_i$ denotes the $i$th standard basis vector.
Differentiating \eqref{eq:second-Q-integral} and transposing yields
\begin{equation}
\begin{aligned}
 dQ_t
 &=\left\{\left(\partial a(X_t^{[0]})-\frac{\epsilon^2}{\sigma^2}\gamma(t)H^\top H\right)Q_t
       +Q_t \partial a(X_t^{[0]})^\top\right\}dt\\
 &\quad+2\gamma(t)\bigl(\partial^2a(X_t^{[0]})[\mu_t]\bigr)^\top dt\\
 &\quad+\frac{\epsilon}{\sigma^2}\bigl(K(t)[H^\top d\nu_t]\bigr)^\top,
\end{aligned}
 \label{eq:second-Q-equation}
\end{equation}
The $Q_t \partial a(X_t^{[0]})^\top$ term comes from differentiating $\Phi(s,t)$;
the left-multiplication term comes from differentiating the endpoint covariance.

\paragraph{The equation for $K(t)$ and closure}
Both $\Phi(s,t)$ and $\gamma(s,t;t)$ have finite variation in $t$.
Thus, for fixed vectors $u,v$, differentiation of
\eqref{eq:second-K-integral} gives
\begin{equation}
\begin{aligned}
 \frac{d}{dt}\{K(t)[u\otimes v]\}
 &=\partial a(X_t^{[0]})K(t)[u\otimes v]
   \\
 &\quad+K(t)\left[\left\{\left(\partial a(X_t^{[0]})^\top-\frac{\epsilon^2}{\sigma^2}H^\top H\gamma(t)\right)u\right\}\otimes v\right]\\
 &\quad+K(t)\left[u\otimes\left\{\left(\partial a(X_t^{[0]})^\top-\frac{\epsilon^2}{\sigma^2}H^\top H\gamma(t)\right)v\right\}\right]
   \\
 &\quad+2\partial^2a(X_t^{[0]})[(\gamma(t)u)\otimes(\gamma(t)v)].
\end{aligned}
 \label{eq:second-K-equation}
\end{equation}
Equivalently, regarding $K(t)$ and $\partial^2a(X_t^{[0]})$ as linear maps
from $\mathbb R^{d_1}\otimes\mathbb R^{d_1}$ to $\mathbb R^{d_1}$, we obtain the tensor equation
\begin{equation}
\begin{aligned}
 \dot K(t)
 &=\partial a(X_t^{[0]})K(t)
   \\
 &\quad+K(t)\left\{\left(\partial a(X_t^{[0]})^\top-\frac{\epsilon^2}{\sigma^2}H^\top H\gamma(t)\right)\otimes I_{d_1}\right\}\\
 &\quad+K(t)\left\{I_{d_1}\otimes\left(\partial a(X_t^{[0]})^\top-\frac{\epsilon^2}{\sigma^2}H^\top H\gamma(t)\right)\right\}\\
 &\quad+2\partial^2a(X_t^{[0]})(\gamma(t)\otimes\gamma(t)).
\end{aligned}
 \label{eq:second-K-tensor-equation}
\end{equation}
Here juxtaposition denotes composition of linear maps, with
$(M\otimes N)(u\otimes v)=(Mu)\otimes(Nv)$.
In particular, for a matrix $M\in\mathbb R^{d_1\times d_1}$,
left multiplication acts on the output of $K(t)$:
\[
 (MK(t))[u\otimes v]=M\bigl(K(t)[u\otimes v]\bigr),\qquad
 (MK(t))_{\ell ij}=\sum_{r=1}^{d_1}M_{\ell r}K_{rij}(t).
\]
No new auxiliary variable is generated, so
\eqref{eq:second-m-equation}, \eqref{eq:second-Q-equation}, and
\eqref{eq:second-K-equation} form a closed system,
with $m_0^{(2)}=0$, $Q_0=0$, and $K(0)=0$.

\begin{remark}\label{rem:random-initial-coefficients}
One can assume an initial expansion
$X_0^\epsilon=x_0+\epsilon Z_0+\frac{\epsilon^2}{2}U_0+O_{L^p}(\epsilon^3)$, where $x_0$ is constant and $Z_0$ is Gaussian. Adding the initial-value terms to the integral representations gives the
same differential equations, initialized by the corresponding joint moments:
$\mu_0=E[Z_0]$, $\gamma(0)=\operatorname{Cov}(Z_0)$,
$m_0^{(2)}=E[U_0]$, $Q_0=\operatorname{Cov}(Z_0,U_0)$, and
\[
 K_{\ell ij}(0)=E[(U_{0,\ell}-E[U_{0,\ell}])
                    (Z_{0,i}-E[Z_{0,i}])(Z_{0,j}-E[Z_{0,j}])].
\]
\end{remark}

\section{Application to the Lorenz--63 model}
\label{sec:lorenz63-specialization}
We now take $d_1=3$ and specialize the drift to the Lorenz--63 vector field
\begin{align}
  a(x)
  =\begin{pmatrix}
    10(x_2-x_1)\\
    x_1(28-x_3)-x_2\\
    x_1x_2-\dfrac{8}{3}x_3
  \end{pmatrix},
\end{align}
where $x=(x_1,x_2,x_3)^\top$.

Write $A(t)=\partial a(X_t^{[0]})$ and $S=\partial^2a(x)$.
The Hessian $S$ is independent of $x$, and
\begin{equation}
 \partial a(x)=\begin{pmatrix}
 -10&10&0\\28-x_3&-1&-x_1\\x_2&x_1&-8/3
 \end{pmatrix},\qquad
 S[u\otimes v]
 =\begin{pmatrix}0\\-u_1v_3-u_3v_1\\u_1v_2+u_2v_1\end{pmatrix}.
 \label{eq:lorenz-derivatives}
\end{equation}
The corresponding partial contraction is
\[
 S[u]=\begin{pmatrix}
 0&0&0\\-u_3&0&-u_1\\u_2&u_1&0
 \end{pmatrix},\qquad S[u]v=S[u\otimes v].
\]
For the numerical implementation, write
\[
 m_t^{(1)}=\mu_t,\qquad P(t)=\gamma(t),\qquad m_t^{(2)}=\tilde E_t[X_t^{[2]}].
\]
The notation for $Q_t$ and $K(t)$ is unchanged. With
$d\nu_t=dY_t^\epsilon-H(X_t^{[0]}+\epsilon m_t^{(1)})\,dt$,
the equations of Section~\ref{sec:second-order-coefficient} specialize as follows.

\begin{paperbox}
\setlength{\abovedisplayskip}{3pt}
\setlength{\abovedisplayshortskip}{0pt}
\setlength{\belowdisplayskip}{3pt}
\setlength{\belowdisplayshortskip}{0pt}
\begin{align*}
 dX_t^{[0]}&=a(X_t^{[0]})\,dt,\\
 dm_t^{(1)}&=A(t)m_t^{(1)}\,dt
          +\frac{\epsilon}{\sigma^2}P(t)H^\top d\nu_t,\\
 \dot P(t)&=A(t)P(t)+P(t)A(t)^\top+I_3
             -\frac{\epsilon^2}{\sigma^2}P(t)H^\top HP(t),\\
 dm_t^{(2)}&=\left\{A(t)m_t^{(2)}
       +2\begin{pmatrix}
         0\\-m_{t,1}^{(1)}m_{t,3}^{(1)}-P_{13}(t)\\m_{t,1}^{(1)}m_{t,2}^{(1)}+P_{12}(t)
        \end{pmatrix}\right\}dt
       +\frac{\epsilon}{\sigma^2}Q_t^\top H^\top d\nu_t,\\
 dQ_t&=\left\{\left(A(t)-\frac{\epsilon^2}{\sigma^2}P(t)H^\top H\right)Q_t
               +Q_tA(t)^\top\right\}dt\\*
     &\quad+2P(t)\bigl(S[m_t^{(1)}]\bigr)^\top dt
       +\frac{\epsilon}{\sigma^2}\bigl(K(t)[H^\top d\nu_t]\bigr)^\top,\\
 \dot K(t)&=A(t)K(t)\\*
     &\quad+K(t)\left\{\left(A(t)^\top
                 -\frac{\epsilon^2}{\sigma^2}H^\top HP(t)\right)\otimes I_3\right\}\\*
     &\quad+K(t)\left\{I_3\otimes\left(A(t)^\top
                 -\frac{\epsilon^2}{\sigma^2}H^\top HP(t)\right)\right\}\\*
     &\quad+2S(P(t)\otimes P(t)).
\end{align*}
\end{paperbox}
Here $m_{t,i}^{(1)}$ denotes component $i$ of $m_t^{(1)}$, and the derivatives
are given explicitly in \eqref{eq:lorenz-derivatives}.
Initially, $X_0^{[0]}=x_0$ and
$m_0^{(1)}=m_0^{(2)}=0$, $P(0)=Q_0=K(0)=0$.
The first- and second-order estimates are
\[
 \widehat X_t^{\epsilon,(1)}=X_t^{[0]}+\epsilon m_t^{(1)},\qquad
 \widehat X_t^{\epsilon,(2)}=X_t^{[0]}+\epsilon m_t^{(1)}+\frac{\epsilon^2}{2}m_t^{(2)}.
\]

\section{Numerical results}
\subsection{Numerical setup and benchmark filters}
The numerical experiments use the Lorenz--63 model and the specialized filter
equations introduced in Section~\ref{sec:lorenz63-specialization}.
Let $t_k=k\Delta t$. We denote the numerical approximations of
$X_{t_k}^\epsilon$ and $Y_{t_k}^\epsilon$ by $X_k^\epsilon$ and
$Y_k^\epsilon$. In all experiments, we use $\Delta t=0.001$ and the
Euler--Maruyama scheme
\begin{align}
  X_{k+1}^\epsilon
  &=X_k^\epsilon+a(X_k^\epsilon)\Delta t
    +\epsilon\sqrt{\Delta t}\,\xi_k,
  \label{eq:numerical-signal}\\
  \Delta Y_k
  &:=Y_{k+1}^\epsilon-Y_k^\epsilon
   =H X_k^\epsilon\Delta t
    +\sigma\sqrt{\Delta t}\,\eta_k,
  \label{eq:numerical-observation}
\end{align}
with $X_0^\epsilon=x_0=(1,1,1)^\top$ and $Y_0^\epsilon=0$. Here
$\{\xi_k\}_{k\ge 0}$ and
$\{\eta_k\}_{k\ge 0}$ are independent i.i.d. sequences with
$\xi_k\sim\mathcal N(0,I_3)$ and
$\eta_k\sim\mathcal N(0,I_{d_2})$. Thus, the observation likelihood for the
discrete model is
\begin{align}
  \Delta Y_k\mid X_k^\epsilon
  \sim
  \mathcal N\!\left(
    H X_k^\epsilon\Delta t,
    \sigma^2\Delta t\,I_{d_2}
  \right).
  \label{eq:numerical-likelihood}
\end{align}
For each seed, all methods use the same truth trajectory and observation
increments. We take $\sigma=0.8$ and consider
$\epsilon\in\{0.2,0.5,1.0\}$. The five observation matrices are
\begin{align*}
  H_{\mathrm{full}} &= I_3, &
  H_{xz} &=
  \begin{pmatrix}
    1 & 0 & 0 \\
    0 & 0 & 1
  \end{pmatrix}, &
  H_x &= \begin{pmatrix} 1 & 0 & 0 \end{pmatrix}, \\
  H_y &= \begin{pmatrix} 0 & 1 & 0 \end{pmatrix}, &
  H_z &= \begin{pmatrix} 0 & 0 & 1 \end{pmatrix}.
\end{align*}
We compare the proposed filter with two benchmark filters that serve different
purposes. A bootstrap particle filter (PF) with many particles provides a
numerical reference for the filtering mean of the discrete model. The ensemble
transform Kalman filter (ETKF) provides a practical finite-ensemble baseline
for the long-time experiment.
Both filters use \eqref{eq:numerical-likelihood}. At step $k$, they assimilate
$\Delta Y_k$ at $t_k$ and then propagate their particles or ensemble members to
$t_{k+1}$ using \eqref{eq:numerical-signal}. We compare the resulting forecast
means with $X_{k+1}^\epsilon$. All random draws used by the PF and ETKF are
independent of the truth and observation noises.
We implement a bootstrap PF based on \citet{gordon1993} with
$M_{\mathrm{PF}}=10000$ particles.
Their initial states are
\begin{align*}
  X_0^{\mathrm p,(i)}=x_0+0.01\,\zeta_0^{\mathrm p,(i)},
  \qquad
  \zeta_0^{\mathrm p,(i)}\overset{\mathrm{i.i.d.}}{\sim}\mathcal N(0,I_3).
\end{align*}
At each step, we update the normalized importance weights using
\eqref{eq:numerical-likelihood}. We apply systematic resampling whenever the
effective sample size
$\mathrm{ESS}_k=(\sum_{i=1}^{M_{\mathrm{PF}}}(w_k^{(i)})^2)^{-1}$ is below
$M_{\mathrm{PF}}/2$.
The PF is used in both the short- and long-time experiments.
For the long-time experiment, we also use the ETKF introduced by
\citet{bishopAdaptiveSamplingEnsemble2001}, in the ensemble-space formulation
of \citet{huntEfficientDataAssimilation2007}. We omit localization and use
$M_{\mathrm{ETKF}}=20$ members. The initial ensemble is
\begin{align*}
  X_0^{\mathrm f,(i)}=x_0+0.01\,\zeta_0^{\mathrm f,(i)},
  \qquad
  \zeta_0^{\mathrm f,(i)}\overset{\mathrm{i.i.d.}}{\sim}\mathcal N(0,I_3).
\end{align*}
The ETKF update used in the experiments is given in
Appendix~\ref{app:etkf-update}.

\subsection{Short-time experiment}

We first examine the short-time behavior of the expansion for the Lorenz--63
model with the time horizon $T=4$. The results are averaged over 100
independent seeds.
Table \ref{tab:t4-rmse} reports the per-component
space--time RMSE against the truth and against the PF estimate (PFRMSE).

\begin{table}[!htbp]
  \centering
  \small
  \caption{Short-time filtering accuracy for \(T=4\).
  RMSE denotes the per-component space--time RMSE of each method against the
  truth. Entries are RMSE (PFRMSE), averaged over 100 seeds; parentheses
  contain the per-component space--time RMSE against the PF estimate with
  10000 particles, not a standard deviation. The PF's distance to itself is
  zero and is denoted by a dash.}
  \label{tab:t4-rmse}
  \begin{subtable}{\textwidth}
    \centering
    \caption{\(\epsilon=0.2\)}
    \begin{tabular}{lcccc}
      \hline
      \(H\) & PF & order0 & order1 & order2 \\
      \hline
      \(H_{\mathrm{full}}\) & 0.281 (---) & 0.314 (0.134) & 0.281 (0.019) & 0.281 (0.013) \\
      \(H_{xz}\)            & 0.304 (---) & 0.328 (0.127) & 0.305 (0.021) & 0.304 (0.014) \\
      \(H_x\)               & 0.291 (---) & 0.296 (0.058) & 0.290 (0.017) & 0.291 (0.008) \\
      \(H_y\)               & 0.288 (---) & 0.296 (0.081) & 0.287 (0.018) & 0.288 (0.010) \\
      \(H_z\)               & 0.288 (---) & 0.296 (0.098) & 0.287 (0.017) & 0.288 (0.011) \\
      \hline
    \end{tabular}
  \end{subtable}

  \vspace{0.5em}
  \begin{subtable}{\textwidth}
    \centering
    \caption{\(\epsilon=0.5\)}
    \begin{tabular}{lcccc}
      \hline
      \(H\) & PF & order0 & order1 & order2 \\
      \hline
      \(H_{\mathrm{full}}\) & 0.542 (---) & 0.791 (0.567) & 0.546 (0.087) & 0.545 (0.056) \\
      \(H_{xz}\)            & 0.618 (---) & 0.830 (0.563) & 0.627 (0.099) & 0.622 (0.060) \\
      \(H_x\)               & 0.674 (---) & 0.746 (0.326) & 0.675 (0.096) & 0.678 (0.044) \\
      \(H_y\)               & 0.628 (---) & 0.746 (0.405) & 0.629 (0.088) & 0.631 (0.049) \\
      \(H_z\)               & 0.607 (---) & 0.746 (0.448) & 0.611 (0.082) & 0.610 (0.046) \\
      \hline
    \end{tabular}
  \end{subtable}

  \vspace{0.5em}
  \begin{subtable}{\textwidth}
    \centering
    \caption{\(\epsilon=1.0\)}
    \begin{tabular}{lcccc}
      \hline
      \(H\) & PF & order0 & order1 & order2 \\
      \hline
      \(H_{\mathrm{full}}\) & 0.806 (---) & 1.634 (1.405) & 0.880 (0.327) & 0.882 (0.306) \\
      \(H_{xz}\)            & 0.938 (---) & 1.726 (1.452) & 1.049 (0.400) & 1.039 (0.368) \\
      \(H_x\)               & 1.159 (---) & 1.533 (1.027) & 1.215 (0.375) & 1.208 (0.271) \\
      \(H_y\)               & 0.997 (---) & 1.533 (1.163) & 1.061 (0.345) & 1.050 (0.271) \\
      \(H_z\)               & 0.945 (---) & 1.533 (1.195) & 1.018 (0.320) & 0.989 (0.237) \\
      \hline
    \end{tabular}
  \end{subtable}
\end{table}

For this short horizon, the zeroth-order approximation is already noticeably
less accurate than the PF reference, while the first- and second-order
approximations are close to PF for all observation patterns.  The change in
truth RMSE from first to second order is modest over these noise levels and
this horizon, whereas the second-order approximation has a smaller PFRMSE
than the first-order approximation in all the reported cases. These results
provide numerical support for the effectiveness of the second-order correction
over the short time horizon considered here.

\subsection{Long-time experiment}
In the long-time application of the expansion, one can expect that the expansion diverges, since the reference solution $X_t^{[0]}$ can deviate substantially from the true state $X_t^\epsilon$ due to the chaotic property of the model. In fact,
Figure \ref{fig:global-expansion-divergence} shows a single-seed example of
the global expansion without recentering.  Although the first-order estimate
remains on the same scale as the zeroth-order estimate, the second-order
estimate eventually grows by several orders of magnitude.  This illustrates
that the global expansion is not suitable for long-time filtering without
additional reparameterization.

\begin{figure}[!htbp]
  \centering
  \includegraphics[width=0.82\textwidth]{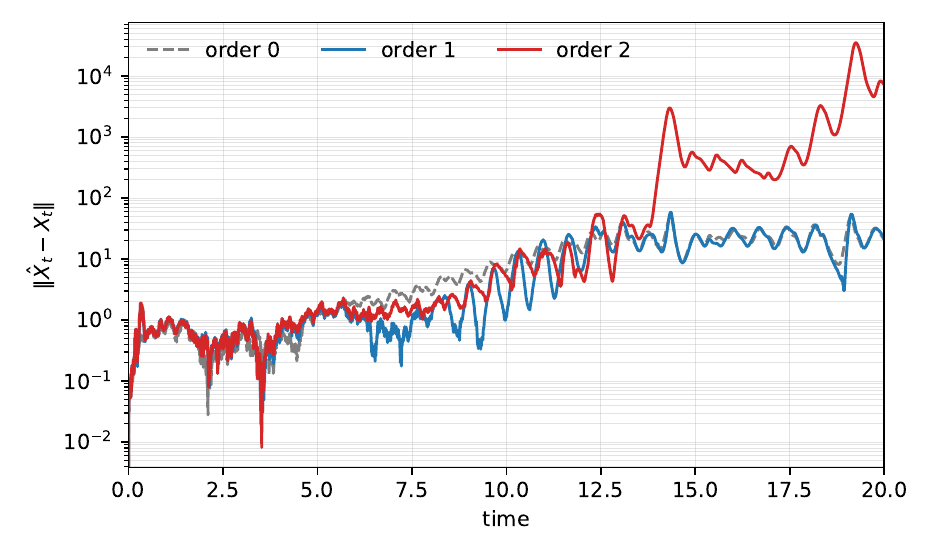}
  \caption{Divergence of the global expansion without recentering for one
  realization of the fully observed Lorenz--63 model.  The plot shows
  \(\|\hat X_t-X_t\|\) for \(T=20\), \(\epsilon=0.5\), and seed 123.}
  \label{fig:global-expansion-divergence}
\end{figure}

To address this issue, we apply the following reset after each time-discretization step.

\begin{paperbox}
\textbf{Reset algorithm}\par\smallskip
Suppressing the time index, let $x$ denote the current reference point,
and let $m^{(1)},\allowbreak m^{(2)},\allowbreak P,\allowbreak Q,\allowbreak K$ denote the current coefficient values.
Apply the following operations in order, using the values before each
operation on its right-hand sides.

\smallskip
\textbf{1. Recentering.}
For the second-order approximation, update
\begin{equation}
  x\leftarrow x+\epsilon m^{(1)}+\frac{\epsilon^2}{2}m^{(2)},
  \qquad m^{(1)}\leftarrow0,\qquad m^{(2)}\leftarrow0,
  \label{eq:reset-centering}
\end{equation}
leaving $P,Q,K$ unchanged. For the first-order approximation, use only
$x\leftarrow x+\epsilon m^{(1)}$ and $m^{(1)}\leftarrow0$, leaving $P$ unchanged,
and omit Step 2.

\smallskip
\textbf{2. Redistribution (second order only).}
Set
\begin{equation}
  L=I_3+\frac{\epsilon}{2}Q^\top P^\dagger,
  \label{eq:reset-linear-map}
\end{equation}
where $P^\dagger$ is the Moore--Penrose pseudoinverse, and update
\begin{equation}
  P\leftarrow LPL^\top,\qquad Q\leftarrow0,\qquad
  K\leftarrow K\times_2 L^\top\times_3 L^\top,
  \label{eq:reset-rule}
\end{equation}
leaving $x,m^{(1)},m^{(2)}$ unchanged. Skip Step 2 if $L$ fails the numerical
nonsingularity check. We then continue the filtering recursion
from the updated coefficients, evaluating the drift and its derivatives at
the updated reference point $x$.
\end{paperbox}

\begin{remark}
(1) The first-order approximation with recentering is similar in construction to the discrete-time extended Kalman filter, in which the model equation is linearized around the posterior mean from the previous step.\\
(2) The redistribution is motivated by viewing the evolution as a composition of short-time
steps, each starting from a random initial state as described in
Remark~\ref{rem:random-initial-coefficients}, conditionally on the observations
at the restart time. After recentering, we rewrite
the second-order expansion of the state as
\begin{align*}
  x+\epsilon Z+\frac{\epsilon^2}{2}U
  &=x+\epsilon\underbrace{\left(I_3+\frac{\epsilon}{2}D\right)}_{L}Z
    +\frac{\epsilon^2}{2}(U-DZ),
\end{align*}
where $D=Q^\top P^\dagger$ and $Z,U$ denote the variables after recentering.
We then evolve the tower equations with the updated initial values
\begin{align*}
  x^+=x,\qquad Z^+=LZ,\qquad U^+=U-DZ,
\end{align*}
so that
\begin{align*}
  (m^{(1)})^+=(m^{(2)})^+=Q^+=0
\end{align*}
at the restart time. This yields the redistribution rule \eqref{eq:reset-rule}.
\end{remark}

We now evaluate the reset procedure over $T=50$, using 100 seeds for each
observation matrix and noise level.

The first-order method uses recentering only, whereas the second-order method
uses both recentering and redistribution after every step. The zeroth-order
method is the deterministic reference trajectory. Table \ref{tab:t50-rmse}
reports RMSE (PFRMSE) together with the
number of failed second-order runs. A run is classified as failed if an
estimate becomes nonfinite, any of its components exceeds $10^6$ in absolute
value, or a numerical error prevents completion.

\begin{table}[!htbp]
  \centering
  \small
  \setlength{\tabcolsep}{4pt}
  \caption{Long-time filtering accuracy for $T=50$. Entries are RMSE
  (PFRMSE), with the same definitions as in Table \ref{tab:t4-rmse}.
  Means are computed over successful runs only. The last column gives the
  number of failed order2 runs out of 100; the other methods completed all
  runs. A standalone dash denotes no successful runs, whereas the dash in
  the PF parentheses denotes its zero distance to itself.}
\label{tab:t50-rmse}
  \begin{subtable}{\textwidth}
    \centering
    \caption{$\epsilon=0.2$}
    \begin{tabular}{lcccccc}
      \hline
      $H$ & PF & ETKF & order0 & order1 & order2 & Failures \\
      \hline
      $H_{\mathrm{full}}$ & 0.723 (---) & 0.911 (0.550) & 10.733 (10.707) & 5.985 (5.934) & 0.745 (0.168) & 0 \\
      $H_{xz}$ & 1.006 (---) & 1.778 (1.446) & 10.688 (10.642) & 8.789 (8.728) & 1.394 (0.839) & 1 \\
      $H_x$ & 1.614 (---) & 2.558 (1.975) & 10.656 (10.534) & 8.595 (8.438) & 3.175 (2.683) & 29 \\
      $H_y$ & 1.074 (---) & 1.536 (1.093) & 10.656 (10.607) & 7.371 (7.293) & 1.588 (1.083) & 11 \\
      $H_z$ & 5.733 (---) & 8.190 (5.986) & 10.656 (9.046) & 9.613 (7.710) & --- & 100 \\
      \hline
    \end{tabular}
  \end{subtable}

  \vspace{0.5em}
  \begin{subtable}{\textwidth}
    \centering
    \caption{$\epsilon=0.5$}
    \begin{tabular}{lcccccc}
      \hline
      $H$ & PF & ETKF & order0 & order1 & order2 & Failures \\
      \hline
      $H_{\mathrm{full}}$ & 0.888 (---) & 0.994 (0.440) & 10.823 (10.785) & 1.533 (1.183) & 0.893 (0.109) & 0 \\
      $H_{xz}$ & 1.183 (---) & 1.708 (1.205) & 10.839 (10.773) & 7.560 (7.461) & 1.325 (0.473) & 0 \\
      $H_x$ & 1.823 (---) & 2.467 (1.648) & 10.790 (10.647) & 8.115 (7.890) & 2.448 (1.554) & 9 \\
      $H_y$ & 1.272 (---) & 1.542 (0.886) & 10.790 (10.724) & 5.022 (4.828) & 1.419 (0.527) & 3 \\
      $H_z$ & 5.909 (---) & 8.081 (5.689) & 10.790 (9.133) & 9.642 (7.598) & --- & 100 \\
      \hline
    \end{tabular}
  \end{subtable}

  \vspace{0.5em}
  \begin{subtable}{\textwidth}
    \centering
    \caption{$\epsilon=1.0$}
    \begin{tabular}{lcccccc}
      \hline
      $H$ & PF & ETKF & order0 & order1 & order2 & Failures \\
      \hline
      $H_{\mathrm{full}}$ & 1.065 (---) & 1.130 (0.370) & 11.079 (11.017) & 1.216 (0.580) & 1.070 (0.101) & 0 \\
      $H_{xz}$ & 1.427 (---) & 1.771 (1.028) & 10.971 (10.879) & 5.347 (5.148) & 1.476 (0.354) & 0 \\
      $H_x$ & 2.080 (---) & 2.532 (1.426) & 11.046 (10.845) & 6.495 (6.155) & 2.279 (0.905) & 1 \\
      $H_y$ & 1.488 (---) & 1.661 (0.729) & 11.046 (10.942) & 2.134 (1.500) & 1.522 (0.325) & 0 \\
      $H_z$ & 6.268 (---) & 8.358 (5.574) & 11.046 (9.118) & 9.493 (7.379) & --- & 100 \\
      \hline
    \end{tabular}
  \end{subtable}
\end{table}

As a practical finite-ensemble baseline, the ETKF completed all runs across
the observation configurations and noise levels considered here. Despite
using only $M_{\mathrm{ETKF}}=20$ ensemble members, compared with
$M_{\mathrm{PF}}=10000$ particles for the PF, it attained truth RMSE close to
that of the PF under full observation.
Its lower truth RMSE than the first-order approximation suggests a practical
benefit of propagating the ensemble with the original nonlinear Lorenz--63
drift, although the observed difference cannot be attributed to this effect
alone. Similar truth RMSE values for the PF and ETKF do not imply a small
PF--ETKF RMSE. Such a discrepancy is consistent with the fact that the ETKF
is not generally expected to converge to the nonlinear filtering posterior
even in the infinite-ensemble limit.

Under full observation, the second-order method completes all runs and
achieves truth RMSE close to that of the PF for all three noise levels, with
substantially smaller PFRMSE than the first-order method with recentering
alone. For example, at $\epsilon=0.5$, the first- and second-order PFRMSEs
are $1.183$ and $0.109$, respectively. The results for $H_{xz}$ also show
markedly smaller errors for the second-order method, although one run fails
at $\epsilon=0.2$. Compared with the ETKF baseline with 20 members,
the second-order method under full observation reduces the mean truth RMSE
by $18.2\%$, $10.1\%$, and $5.4\%$ for $\epsilon=0.2,0.5,1.0$,
respectively. The corresponding reductions in mean PFRMSE are $69.6\%$,
$75.2\%$, and $72.6\%$. These percentages are computed from the unrounded
means as $100(1-\mathrm{error}_{\mathrm{order2}}/\mathrm{error}_{\mathrm{ETKF}})$.

The improvement is less robust under one-dimensional observations. The
second-order method fails in some runs for $H_x$ and $H_y$, and in all runs
for $H_z$, whereas the ETKF completes all runs. Where failures occur, the
second-order method's reported errors are conditional on
successful completion and are not directly comparable to the other methods'
100-seed averages. Even among successful runs, the tabulated mean errors
are not uniformly smaller than the ETKF averages: at $\epsilon=0.2$,
both errors are larger for $H_x$, and truth RMSE is larger for $H_y$.
Thus the reset procedure enables accurate long-time
filtering in the fully observed examples, but does not prevent numerical
breakdown for all observation patterns.

To isolate the role of redistribution, Figure \ref{fig:recentering-ablation}
compares the second-order method with recentering alone against the method
with both reset operations, using the same seed and observation setting as
Figure \ref{fig:global-expansion-divergence} and extending the same realization
to $T=50$. Recentering alone delays the growth of the error but does not
prevent numerical breakdown at $t\approx29.65$. With redistribution, the error remains
small over the displayed interval. This example illustrates the importance
of redistribution beyond recentering, although it does not establish
stability for all realizations.

\begin{figure}[!htbp]
  \centering
  \includegraphics[width=0.82\textwidth]{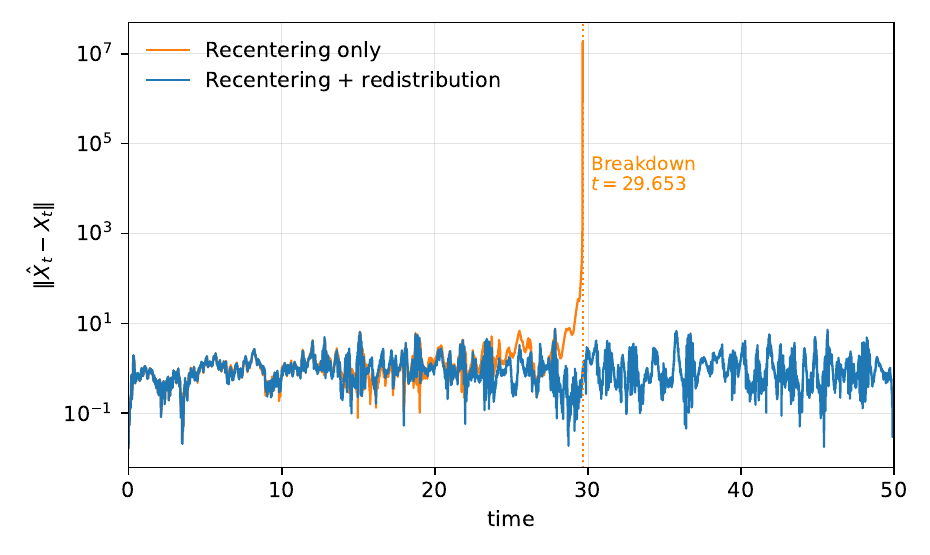}
  \caption{Effect of redistribution on the second-order approximation under
  full observation, with $\epsilon=0.5$, seed 123, $\Delta t=0.001$, and
  $\sigma=0.8$. Both methods apply recentering after every step; only the blue
  curve additionally uses redistribution. The recentering-only curve is
  stopped when an estimate component first exceeds $10^6$ in absolute value
  ($t=29.653$).}
  \label{fig:recentering-ablation}
\end{figure}

\FloatBarrier
\section{Summary and discussion}
In this work, we derived recursive formulas for computing the asymptotic expansion of the conditional expectation. In particular, the zeroth-order approximation is the deterministic reference solution, the first-order coefficient is obtained from a Kalman--Bucy filter for the linearized model, and the second-order coefficient is computed through a closed finite-dimensional system involving the auxiliary variables $Q$ and $K$.
This finite-dimensional structure extends to higher-order coefficients.

In the short-time experiments for the Lorenz-63 model with time horizon $T=4$, we observed that including the second-order correction reduced the discrepancy from the particle-filter reference with 10,000 particles in all the cases considered. These results provide numerical support for the accuracy of the expansion as an approximation to the conditional expectation over short time horizons.

Over longer time horizons, however, simply applying the expansion is insufficient to maintain accurate estimates. To stabilize its long-term behavior, we introduce recentering and redistribution at each Euler step. With these reset procedures, the second-order approximation improves on the recentered first-order approximation under full and $xz$-observations in the reported experiments. In particular, under full observation, it achieves accuracy close to that of the particle-filter reference with 10,000 particles over the time horizon $T=50$. 

In the Lorenz--63 example, the recursive state, including the reference solution, requires only 42 independent scalar components when the symmetries of $P$ and $K$ are taken into account. For a signal of dimension $d_1$, the number of independent scalar components required up to second order is
\begin{align*}
  3d_1+\frac{d_1(d_1+1)}{2}+d_1^2+\frac{d_1^2(d_1+1)}{2}.
\end{align*}
Thus, the dimension of the second-order coefficient system grows polynomially with the state dimension, whereas standard particle filters can require exponentially many particles to avoid weight degeneracy in high-dimensional settings \citep{snyderObstaclesHighDimensionalParticle2008}. Together with the accuracy observed in our Lorenz--63 experiments, this suggests the potential to achieve comparable filtering accuracy at a lower computational cost than standard particle filters in higher-dimensional problems.

Nevertheless, the second-order approximation can still become unstable even with these resets when fewer state components are observed. Furthermore, a rigorous explanation of how repeated resets preserve accuracy over long time horizons remains to be established, and a stable extension of the redistribution procedure to third- and higher-order approximations is still an open problem. These limitations indicate that the current reset procedure is not yet theoretically complete and leave room for further improvement.

For applications to higher-dimensional models, a further reduction in the dimension of the coefficient system is required. Although the number of variables needed for an expansion of fixed order grows polynomially with the state dimension, the resulting system is still too large for operational weather forecasting models.

However, the covariance matrix governed by the Riccati equation can itself exhibit an intrinsic low-rank structure. For linear dynamics without process noise, \citet{bocquet2017degenerate} show that the Kalman filter covariance is asymptotically supported on the unstable-neutral subspace, whose dimension can be much smaller than the full state dimension. Numerical studies of the EnKBF for Lorenz--63 and the ETKF for Lorenz--96 also show that accurate state estimation is possible with rank-deficient ensemble covariances in nonlinear models \citep{dewiljesLongTimeStability2018,takedaNoisescaledAccuracyEnsemble2026}. Motivated by these results, we expect that the auxiliary matrices and tensors in our coefficient system may also exhibit low effective rank in suitable regimes, allowing the system to be compressed through mathematical techniques that exploit this structure. Establishing such a property in the presence of system noise and developing the corresponding reduced equations remain topics for future work.

\section*{Acknowledgments}
The authors thank Professor Yasumasa Matsuda for helpful discussions and advice.
This work was supported by JSPS KAKENHI Grant Number JP26H01952.
Masahiro Kurisaki was supported by the RIKEN Special Postdoctoral Researcher Program.
Kota Takeda was supported by JSPS KAKENHI Grant Number JP26K17033.

\appendix
\section{ETKF update used in the numerical experiments}
\label{app:etkf-update}

Write $M=M_{\mathrm{ETKF}}$ for the ETKF ensemble size, and let
$\{X_k^{\mathrm f,(i)}\}_{i=1}^M$ be the forecast ensemble at $t_k$
before assimilating $\Delta Y_k$, and define its mean and unnormalized anomaly
matrix by
\begin{align*}
  \bar X_k^{\mathrm f}
  &=\frac{1}{M}\sum_{i=1}^M X_k^{\mathrm f,(i)},\\
  \delta X_k^{\mathrm f}
  &=\begin{bmatrix}
    X_k^{\mathrm f,(1)}-\bar X_k^{\mathrm f} & \cdots &
    X_k^{\mathrm f,(M)}-\bar X_k^{\mathrm f}
  \end{bmatrix}.
\end{align*}
Using the likelihood \eqref{eq:numerical-likelihood}, we set
\begin{align*}
  S_k
  &=(M-1)I_M
    +\frac{\Delta t}{\sigma^2}
      (H\delta X_k^{\mathrm f})^\top H\delta X_k^{\mathrm f},\\
  \bar X_k^{\mathrm a}
  &=\bar X_k^{\mathrm f}
    +\frac{1}{\sigma^2}
      \delta X_k^{\mathrm f}S_k^{-1}
      (H\delta X_k^{\mathrm f})^\top
      \left(\Delta Y_k-\Delta t\,H\bar X_k^{\mathrm f}\right),\\
  T_k&=\left[(M-1)S_k^{-1}\right]^{1/2},\\
  X_k^{\mathrm a,(i)}
  &=\bar X_k^{\mathrm a}
    +\delta X_k^{\mathrm f}T_k e_i,
  \qquad i=1,\ldots,M,
\end{align*}
where $T_k$ is the unique symmetric positive-definite square root and $e_i$ is
the $i$th standard basis vector of $\mathbb R^M$. Each analysis member is then
advanced by one Euler--Maruyama step,
\begin{align*}
  X_{k+1}^{\mathrm f,(i)}
  =X_k^{\mathrm a,(i)}
   +a(X_k^{\mathrm a,(i)})\Delta t
   +\epsilon\sqrt{\Delta t}\,\omega_k^{(i)},
\end{align*}
where $\{\omega_k^{(i)}\}_{k\geq0,\,1\leq i\leq M}$ are i.i.d.
$\mathcal N(0,I_3)$ vectors. The reported ETKF estimate is
\begin{align*}
  \widehat X_{k+1}^{\mathrm{ETKF}}
  =\frac{1}{M}\sum_{i=1}^M X_{k+1}^{\mathrm f,(i)}.
\end{align*}

\bibliographystyle{plainnat}
\bibliography{order2_lorenz63_filtering}

\end{document}